\documentclass[journal,twocolumn,10pt]{IEEEtran}
\ifCLASSINFOpdf
\else
\fi
\newtheorem{theorem}{\textbf{Theorem}}

\newtheorem{proof}{\textbf{Proof}}
\newtheorem{lemma}{\textbf{Lemma}}

\usepackage{algorithmic}
\usepackage{array}
\usepackage{color}
\usepackage{cite}
\usepackage[mathscr]{eucal}
\usepackage{graphics}
\usepackage{float}
\usepackage{cuted}
\usepackage{mathtools}
\usepackage{amssymb}
\usepackage{mathrsfs}
\usepackage{mathdots}
\usepackage{amsmath}

\usepackage[ruled]{algorithm2e}
\usepackage{pgfplots}
\pgfplotsset{compat=1.17}
\usepackage{epstopdf}
\usepackage{tikz}
\usepackage{stfloats}
\usepackage{graphicx}
\usepackage{booktabs}
\usepackage{multirow}
\usepackage{listings}
\usepackage{amsfonts}
\usepackage{bm}
\usepackage{bbm}

\usepackage{hyperref}
\usepackage{marvosym}

\usepackage{subfigure}
\graphicspath{{figures/}}

\begin{document}
%
\title{Channel-Aware Ordered Successive Relaying with Finite-Blocklength Coding  }
%
%
%
\author{Lingrui~Zhang, Yuxing~Han, Qiong~Wang, and Wei~  Chen,~\IEEEmembership{Senior Member,~IEEE}%
\thanks{Lingrui~Zhang and Wei~Chen are with the Department of Electronic Engineering and Beijing National Research Center for Information Science and Technology, Tsinghua University, Beijing 100084, China, e-mail: \href{mailto:zlr15@tsinghua.org.cn}{zlr15@tsinghua.org.cn }; \href{mailto:wchen@tsinghua.edu.cn}{wchen@tsinghua.edu.cn}.}%
\thanks{Yuxing~Han is with the Shenzhen International Graduate School, Tsinghua University, e-mail: \href{mailto:yuxinghan@sz.tsinghua.edu.cn}{yuxinghan@sz.tsinghua.edu.cn}.}%
\thanks{Qiong~Wang is with the State Grid Beijing Electric Power Company, e-mail: \href{mailto:wangqiong@bj.sgcc.com.cn}{wangqiong@bj.sgcc.com.cn}.}
\thanks{This work has been submitted to the IEEE for possible publication. Copyright may be transferred without notice, after which this version may no longer be accessible.}
}

%
%

\markboth{IEEE TRANSACTIONS ON
VEHICULAR TECHNOLOGY}%
{Submitted paper}
%



\maketitle

\begin{abstract}
Successive relaying can improve the transmission rate by allowing
the source and relays to transmit messages simultaneously, but it may cause severe inter-relay interference (IRI). IRI cancellation schemes have been proposed to mitigate IRI. However, interference cancellation methods have a high risk of error propagation, resulting in a severe transmission rate loss in finite blocklength regimes. Thus, jointly decoding for successive relaying with finite-blocklength coding (FBC) remains a challenge. In this paper, we present an optimized channel-aware ordered successive
relaying protocol with finite-blocklength coding (CAO-SIR-FBC), which can recover the rate loss by carefully adapting the relay transmission order and rate. We analyze the average throughput of the CAO-SIR-FBC method, based on which a closed-form expression in a high signal-to-noise regime (SNR) is presented. Average throughput analysis and simulations 
show that CAO-SIR-FBC outperforms conventional two-timeslot half-duplex relaying in terms of spectral efficiency.
\end{abstract}


%
\IEEEpeerreviewmaketitle

\section{Introduction}

With the rapid development of 5G and 6G wireless networks,
ultra-reliable low-latency communication (URLLC) systems have attracted considerable attention for their potential applications in telesurgery, smart cities and autonomous driving. 3rd Generation Partnership Project (3GPP) 
researchers from academia and industry have begun to look beyond 5G and have started researching 6G systems for multiple use cases, such as artificial intelligence, augmented/virtual reality (AR/VR), and automatic driving \cite{yang20196g}. URLLC systems have attracted considerable attention because they support the stringent requirements of 6G systems. Specifically, URLLC has strict standards requiring $10^{-10}$ error probability and $1$ ms latency time \cite{20196gwirelessnetwork}. To achieve these strict requirements, supporting a low delay violation
probability has attracted significant attention\cite{2019effective}.

As one of three main usage scenarios for 5G systems, URLLC still relies on orthogonal frequency division multiplexing (OFDM) as 4G does \cite{20173GPP170379}. Different from 4G systems, the subcarrier spacing is allowed to be enlarged \cite{20173GPP38.913}. Due to the enlarged subcarrier spacing, the duration of one timeslot is diminished \cite{20183GPP38.300}. Accordingly, the latency of the time-slotted scheduling strategy is also diminished.
However, the submillisecond duration of a timeslot results in a data packet blocklength that is too short \cite{2016toward}.  Shannon's capacity given in \cite{1948mathematical} is no longer achievable with the short packet. Consequently, short-packet transmissions are adopted in URLLC using FBC coding approaches. In a finite blocklength regime, the maximal transmission rate over additive white Gaussian noise (AWGN) channels
is a decrease function of error probability $\varepsilon$ and an  increase function of blocklength $n$. Moreover, transmission rate analysis has been extended to other channel models, such as multiple-antenna fading
channels \cite{2014quasi}, multiple-antenna Rayleigh fading channels \cite{2015short}, block fading channels \cite{2013block}, and
hybrid automatic repeat request (HARQ) \cite{2018low}.

Satisfying the reliable requirement of URLLC with FBC is a challenge because there is a fundamental tradeoff between latency and reliability \cite{2020joint}. More specifically, a short blocklength causes a severe loss in coding gain \cite{2018short}. Because of the larger subcarrier spacing, the number of subcarriers is reduced under the same bandwidth, which decreases the maximum frequency diversity gain in OFDM systems. Hence, there is a need to reuse the spatial domain to ensure reliability in URLLC. This needs to be done both at each network node and by densifying the entire network \cite{20185gdesign}. At the node level, multiple-input multiple-output (MIMO) systems are adopted in URLLC because they achieve very large capacity increases and antenna diversity gains over a harsh link \cite{1999capacity}. Combined with OFDM, a MIMO-OFDM system can achieve higher capacity and reliability \cite{2006MIMOOFDM}. However, the use of MIMO
technology may not be practical if nodes in a network are small, inexpensive, or typically have severe energy constraints \cite{2008unified}. Cooperative communication is a solution to overcome this limitation by which
a high transmission rate is achieved over harsh wireless links without causing high implementation complexity. Cooperative communication was first proposed by Sendonaris
$et \ al.$ for Code Division Multiple Access systems \cite{sendonaris2003user, sendonaris2003user2}. Since then, much research has focused on cooperative communication. The space diversity and diversity-multiplexing tradeoff  of various cooperative relaying schemes were studied by  Laneman $et \ al.$  \cite{laneman2004cooperative}. These studies showed that the amplify-and-forward (AF) and the selective decode-and-forward (DF)
protocols can achieve the maximal diversity gain over fading channels \cite{nabar2004fading}. The maximal diversity gain of these protocols is equal to the number of employed nodes in cooperative transmission. Most of the research mentioned above is focused on two-timeslot relaying protocols. However, these relaying protocols suffer from a half-duplex constraint, which causes the multiplexing gains of two-timeslot relaying protocols to be upper-bounded by $\frac{1}{2}$.
To overcome the multiplexing loss, some cooperative relaying networks with full-duplex relays have been studied in recent studies \cite{li2016multi}. Full-duplex relays can receive and forward signals at the same time. However, the self-interference caused at full-duplex relays still makes it challenging to implement in practice.

In addition to using full-duplex relays, successive relaying is also a
feasible method of recovering the multiplexing loss \cite{yang2007towards}. The basic idea relies on the current transmission of the source and relays. In this protocol, the cost of sending one message is less than two timeslots.
Therefore, the multiplexing gain of successive relaying
is greater than 1/2, while the half-duplex constraint is
satisfied.
Unfortunately, successive relaying may cause severe IRI, resulting in a low transmission rate and poor reliability. For DF-based successive relaying, IRI causes severe decoding errors and error propagation.
Therefore, mitigating IRI becomes a vital issue for DF-based successive relaying.
Relay scheduling or selection was applied to improve the throughput of DF-based two-path successive relaying in \cite{nomikos2012successive}. Hu $et \ al.$ investigated an IRI mitigation method in which two signals transmitted by the source and relay are jointly decoded in a multiple-relay cooperative network\cite{hu2012efficient}. In \cite{wei2020successive}, a DF relaying protocol was proposed based on an analog network interference cancellation method with linear processing
without decoding the signals from the source relay. A channel-aware successive relaying protocol (CAO-SIR) was proposed in \cite{2014CAO}, where the decode-and-cancel protocol
is extended to cancel the IRI completely.

Recently, researchers have investigated relaying under a finite blocklength regime. The authors in \cite{2019finite} derived closed-form expressions for the coding rates of relay communications under the Nakagami-m fading channel in a finite blocklength regime. In \cite{2016finite}, Du $et \ al.$ studied the relaying throughput performance of multi-hop
relaying networks with FBC under quasistatic
Rayleigh channels.
However, a successive relaying protocol in a finite blocklength regime suffers from a severe loss of transmission rate because of the high risk of error propagation caused by the interference cancellation method of the network. Thus, achieving low-latency short-packet communications is still an open problem for successive relaying protocols.

In this paper, we investigate a DF-based successive relaying protocol, also referred to as CAO-SIR-FBC. The IRI cancellation method is based on CAO-SIR \cite{2014CAO}. More specifically, an interfered relay obtains prior knowledge transmitted from the source, which is used to mitigate IRI without high computational complexity. We further analyze the influence of FBC. The error probabilities of every decoding process are not ignored in this context. The error probability at the destination, which is usually constrained to guarantee the high-reliability requirement of the system, is determined using error-propagation theory. This observation motivates us to further optimize the CAO-SIR-FBC scheme. Specifically, the transmission rate in every timeslot is optimized according to the error probabilities in the decoding processes. However, it is not trivial to
derive the closed-form expressions of the transmission rate with an arbitrary SNR. Based on the optimized CAO-SIR-FBC, we present the average throughput of our scheme.
Compared to classic protocols, the average throughput is increased in our scheme. Furthermore, CAO-SIR-FBC outperforms conventional protocols in terms of the multiplexing gain.

The rest of this paper is organized as follows. Section \ref{sec:System Model} presents the system model. The CAO-SIR-FBC method is presented along with its algorithm in Section \ref{sec:CAO-SIR-FBC}. The performance analysis of the CAO-SIR-FBC method is carried out in Section \ref{sec:PERFORMANCE ANALYSIS}. Finally, the numerical results and conclusions are presented in Sections \ref{ sec:NUMERICAL RESULTS} and \ref{sec:Conclusions }, respectively.

\section{System Model}\label{sec:System Model}

\begin{figure}[t]
	\centering
	\includegraphics[width=0.45\textwidth]{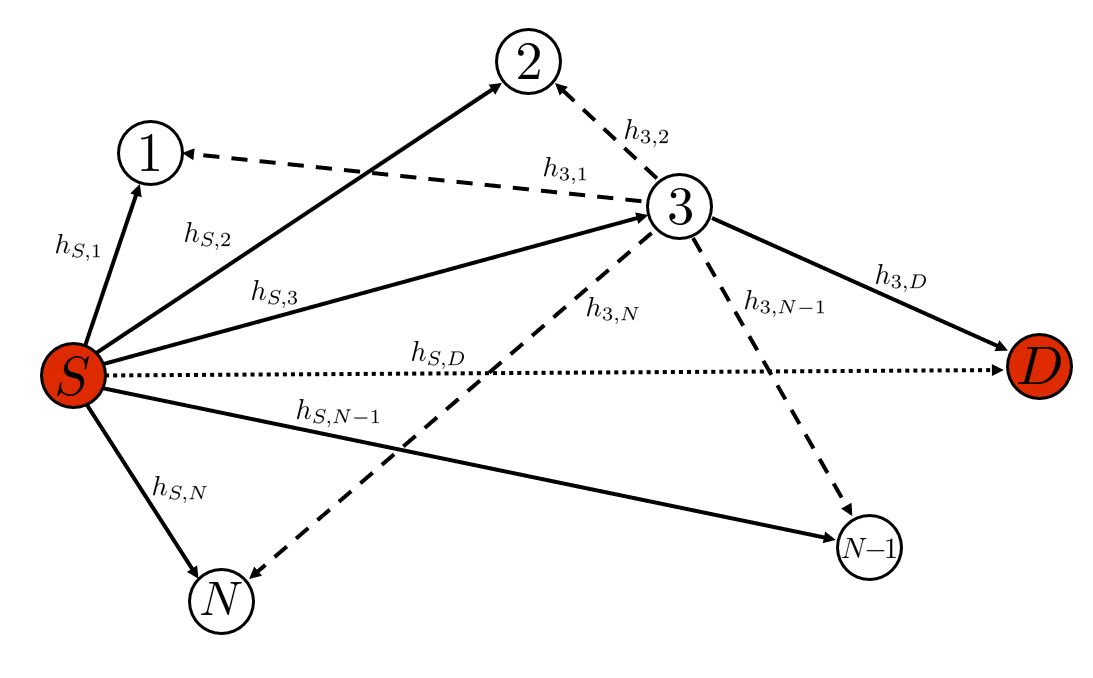}
\caption{System model}
	\label{fig:system model}
\end{figure}

We consider a cooperative communication system in which a source node $S$ communicates to a destination node $D$ with $N$ relays denoted by $\mathcal{N}=\{1, 2, \cdots, N\}$, as shown in Fig. \ref{fig:system model}.
The channel coefficients of the links between nodes $a$ and $b$
are denoted by $h_{a, b}$, where $a\in \{S, 1, \cdots, N\}$ and $b\in \{1, \cdots, N, D\}$.
We use $g_{a,b}$ to denote the channel gain of links between nodes $a$ and $b$, namely, $g_{a,b}=|h_{a,b}|^2$.

We consider a quasi-static or slow-fading channel, in which
the channel state remains constant in each successive relaying period.
In addition, channel estimation with pilot and channel-state-information (CSI) feedback are adopted.
Specifically, the source and relay nodes
send pilots at the beginning of each successive
relaying period.
After receiving the broadcast pilots, a relay can estimate the channel coefficients of the source-relay and relay-relay links.
Similarly, the destination node $D$ can estimate the channel coefficients of the source-destination and relay-destination links.
Consequently, $D$ has the CSI of source-destination and relay-destination channels. A relay has the CSI of links from $S$ and other relays to itself.
Finally, the source-relay, source-destination, and relay-destination channel gains are fed back to the source $S$ through a signaling channel.

The half-duplex constraint is assumed throughout this paper.
A relay node cannot transmit and receive at the same time.
Time-slotted scheduling is assumed in this context.
We use $X_a[k]$ to denote the signal transmitted by node $a$ in the $k$th timeslot.
Node $b$, which does not transmit in the $k$th timeslot because of the half-duplex constraint, receives a signal denoted by
\begin{equation}
	Y_b[k]=\sum_{a\in \mathfrak{A} [k]} h_{a, b}X_a[k]+Z_b[k],
\end{equation}
where $\mathfrak{A} [k]$  denotes the set of transmitting nodes in the $k$th timeslot.
Notation $Z_b[k]$ denotes the additive white Gaussian noise (AWGN) at node $b$, which is subject to a normal distribution with zero mean and a variance of $\sigma ^2$,
i.e., $ Z_b[k] \sim \mathcal{CN} (0,  \sigma^2)$.
The transmission power of node $a$ is denoted by $P_a=E\{| X_a[k]\vert^2 \}$. We assume that each node has the same
transmission power $P$, i.e., $P=E\{| X_a[k]\vert^2 \}$ for all $a\in \{S, 1, 2, \cdots, N\}$.
The transmitter-side SNR of node $a$ is denoted by $\gamma=\frac{P}{\sigma^2}$. Hence, the receiver-side SNR of the link between node $a$ and node $b$ is denoted by $\rho_{a,b}    =g_{a,b}\gamma$.

Because our research is based on a CAO-SIR model, we follow the relay order and scheduling schemes of basic CAO-SIR
\cite{2014CAO}. More specifically, the channel gains of the relay set satisfy
\begin{equation}
	g_{S, 1}\leq g_{S, 2}\leq\cdots\leq g_{S, N}.
\end{equation}
The received signals of relay $m$, $m=1, 2\cdots, N$, are given by
\begin{align}
	Y_m[1]= & h_{S, m}X_S[1]+Z_m[1]\label{eq:Y_m 1}, \\
	\begin{split}
		Y_m[k]=&h_{S, m}X_S[k]+h_{k-1, m}X_{k-1}[k]\label{eq:Y_m k}
		\\&+Z_m[k], \qquad k=2, \cdots, m .
	\end{split}
\end{align}
At the destination node $D$, the received signals
are presented by
\begin{align}
	Y_D[1]=   & h_{S, D}X_S[1]+Z_D[1],                      \\
	\begin{split}
		Y_D[k]=&h_{S, D}X_S[k]+h_{k-1, D}X_{k-1}[k]\label{eq:Y_D k}
		\\&+Z_D[k], \qquad k=2, \cdots, N,
	\end{split}                               \\
	Y_D[N+1]= & h_{N, D}X_N[N+1]+Z_D[N+1]\label{eq:Y_D N+1}.
\end{align}
The IRI cancellation mechanism at relay $m$ is shown by
\begin{equation}\label{Eq:cancel IRI}
	Y_m[i+1]-h_{i, m}X_i[i+1]=h_{S, m}X_S[i+1]+Z_m[i+1].
\end{equation}
The right-hand side of Eq. (\ref{Eq:cancel IRI}) is equivalent to an AWGN channel without IRI. Hence, relay $m$ can decode $W_m$.
Similarly, node $D$ is capable of thoroughly mitigating the interference of $S$ in Eq. (\ref{eq:Y_D k}) as
\begin{equation}\label{Eq: D cancel Interference}
	Y_D[i]-h_{S, D}X_S[i]=h_{i-1, D}X_{i-1}[i]+Z_D[i].
\end{equation}

Finite-blocklength coding is assumed.
Specifically, a packet consisting of $k$ payload bits is typically encoded into a complex symbol whose blocklength is $n$.
The ratio is denoted by $	R=k/n $,
i.e., the number of information bits per complex symbol represents the transmission rate. Let $R^*(\varepsilon)$ denote the largest transmission rate $k/n$ of a coding scheme whose packet error probability does not exceed $\varepsilon$. Additionally, we assume that the blocklength $n$ is constant during
the entire transmission. Therefore, the normal approximate rate given in \cite{polyanskiy2010channel} is presented by
\begin{equation}\label{eq:R usual}
	R^*(\varepsilon)=C-\sqrt{\frac{V}{n}}Q^{-1}(\varepsilon)+\frac{\log n}{2n} .
\end{equation}
Here,
$Q^{-1}(\cdot)$ denotes
the inverse of the Gaussian $Q$ function
(the tail distribution function of the
standard normal distribution).\footnotemark
\footnotetext{
As usual,
$Q(x) = \int_x^\infty \frac{1}{\sqrt{2\pi}} \mathrm{e}^{-t^2/2}dt.$
	\label{eq:Q function}
}
Notation $V$ is
the channel dispersion \cite{polyanskiy2010channel}.

Regarding the case of a real AWGN
channel with the receiver-side SNR $\rho$,
the capacity  is given by \cite{2016toward}

\begin{align}
	C(\rho) & =\frac{1}{2} \log (1+\rho)    .       \label{eq:C}       
\end{align}
The channel dispersion is expressed as
\begin{equation}
    V(\rho)  =\frac{\rho}{2} \frac{(2+\rho)}{(1+\rho)^2}(\log \mathrm{e})^2 . \label{eq:V}
\end{equation}

\section{Successive DF relaying protocol with FBC}\label{sec:CAO-SIR-FBC}

In this section, we propose
the optimization problem of the channel-aware rate adaptation of the CAO-SIR-FBC method. Furthermore, the algorithms to compute the maximal transmission rate are expressed in the second part.

\subsection{Channel-aware rate adaptation}

Because of the FBC, the packet error probability in every timeslot should be considered. We
denote by $\varepsilon_m[k]$  the supremum of the error probability when relay $m$ decodes $W_k$.
When node $D$ decodes message $W_k$, similarly, the supremum of the error probability is denoted by $\varepsilon_D[k]$.
The error probabilities of $W_k$ are dependent.
Based on the definition of the rate,
the rate of message $W_k$ remains constant in the
transmission process where the number of payload bits is constant.
That is,
the rates of transmitting message $W_k$  through
$S$-relay $m$ and relay $k$-$D$ channels are constant
for all $m=k, k+1, \cdots, N$.
Therefore, based on Eq. (\ref{eq:R usual}), the maximal rates of
message $W_k$ in $S$-relay $m$ channels are presented by
\begin{align}\label{eq: R_wk}
		R^*_{k}  
		         =&C_{S, m}
		-\sqrt{\frac{V_{S, m}}{n}}Q^{-1}\left(\varepsilon_m[k]\right)  +\frac{\log n}{2n}                        .   
\end{align}
The maximal rates of message $W_k$ in the relay $k$-$D$ channel are presented by
\begin{align}
		R^*_{k}  =&C_{k, D}
		-\sqrt{\frac{V_{k, D}}{n}}Q^{-1}(\varepsilon_D[k])            
		        +\frac{\log n}{2n} .
\end{align}
Here, $C_{S, m}$ and $C_{k, D}$ denote the capacity of
channels of $S$-relay $m$ and relay $k$-$D$, respectively. We use $V_{S, m}$ and $V_{k,D}$ to denote the capacity of the channels of relay $S$ $m$ and relay $k$-$D$, respectively.

Because all 
messages $W_k$, $k=1, 2\cdots, N$, are reliably transmitted in $N+1$ timeslots,
the average DF-FC-SR rate is presented by
\begin{equation}\label{eq:average R}
	\bar{R}^*=\frac{1}{N+1}\sum\limits _{k=1}^N R^*_{k}.
\end{equation}
Substituting Eq. (\ref{eq:average R}) into Eq. (\ref{eq:R usual}), we have
\begin{equation}
	\begin{split}\label{eq:R average expansion}
		\bar{R}^*=&\frac{1}{N+1}\bigg(\sum\limits_{k=1}^N C_{S, k}
		-\sqrt{\frac{1}{n}}\sum\limits_{k=1}^N \sqrt{V_{S, k}} Q^{-1}(\varepsilon_k[k]) 
		\\
		&+N\cdot \frac{\log n}{2n}\bigg).
	\end{split}\end{equation}

From Eq. (\ref{eq: R_wk}),
all of the error probabilities about $W_k$ are expressed as
\begin{align}
	\begin{split}
		\varepsilon_m[k]=&Q\bigg(
		\sqrt{\frac{n}{V_{S, m}}}\left(C_{S, m}-C_{S, k}\right)\\
		&\qquad+\sqrt{\frac{V_{S, k}}{V_{S, m}}}Q^{-1}\left(\varepsilon_k[k]\right)
		\bigg)\label{eq:epsilon_m[k]},
	\end{split}                                                         \\
	\begin{split}
	\varepsilon_D[k]  =&Q\bigg(
	\sqrt{\frac{n}{V_{k, D}}}(C_{k, D}-C_{S, k})                            \\       
	                 &\qquad+\sqrt{\frac{V_{S, k}}{V_{k, D}}}Q^{-1}(\varepsilon_k[k])
	\bigg)\label{eq:epsilon_D[k]} .
	\end{split}
\end{align}
Next, the reliability of the CAO-SIR-FBC method is given by the following theorem.
\begin{theorem}\label{THEOREM 1}
In the CAO-SIR-FBC method, the reliability, denoted by $\zeta$, is given by
\begin{equation} \label{eq: lower limit of success probability}
    \zeta=\prod_{k=1}^{N}(1-\varepsilon_D[k])
	\prod_{k=1}^{N}\prod_{m=k}^{N}
	(1-\varepsilon_m[k]),
\end{equation}
where $m$ is the relay's order.
\end{theorem}
\begin{proof}
    See Appendix \ref{sec:appendices 4}.
\end{proof}

In a communication system, the error probability of the decoding
process at node $D$ is a constraint of the system.
We suppose that the supremum of the error probability of $D$
is denoted by $\varepsilon_d$. Hence, based on the definition
of $\varepsilon_d$, we have
\begin{equation}\label{ineq: constraint}
	\zeta\geq 1-\varepsilon_d.
\end{equation}

Because only the second item on the right-hand side in Eq. (\ref{eq:R average expansion}) is denoted by $\varepsilon$, we formulate the optimization problem presented by
\begin{equation}\label{opt: original}
\begin{split}
		\min \quad &\sum\limits_{k=1}^N \sqrt{V_{S, k}}
		Q^{-1}(\varepsilon_k[k])\\
		s.t. \quad&\begin{cases}
			\prod\limits_{k=1}^{N}(1-\varepsilon_D[k])
			\prod\limits_{k=1}^{N}\prod\limits_{m=k}^{N}
			(1-\varepsilon_m[k]) \geq 1-\varepsilon_d\\
			\begin{aligned}
				\begin{split}
		\varepsilon_m[k]=&Q\bigg(
		\sqrt{\frac{n}{V_{S, m}}}\left(C_{S, m}-C_{S, k}\right)\\
		&\qquad+\sqrt{\frac{V_{S, k}}{V_{S, m}}}Q^{-1}\left(\varepsilon_k[k]\right)
		\bigg),
	\end{split}                                                         \\
	\begin{split}
	\varepsilon_D[k]  =&Q\bigg(
	\sqrt{\frac{n}{V_{k, D}}}(C_{k, D}-C_{S, k})                            \\       
	                 &\qquad+\sqrt{\frac{V_{S, k}}{V_{k, D}}}Q^{-1}(\varepsilon_k[k])
	\bigg).
	\end{split}
			\end{aligned}
		\end{cases}
	\end{split}\end{equation}

\subsection{Algorithms to compute the maximal transmission rate}
Because the optimization problem (\ref{opt: original}) cannot be solved directly, we need to transform it into an equivalent simpler problem.
We reordered $\varepsilon$
according to their magnitudes.
The re-ordered error probabilities are given by
\begin{equation}
	\varepsilon_1\leq \varepsilon_2\leq \cdots \leq \varepsilon_M.
\end{equation}
Here, $M$, which is the number of $\varepsilon$, is given by
\begin{equation}
	\begin{split}
		M=\frac{N(N+3)}{2}.
	\end{split}
\end{equation}
With the re-ordered $\varepsilon$, Eq. (\ref{eq: lower limit of success probability}) is expressed as
\begin{equation}
	\zeta= \prod_{i=1}^M(1-\varepsilon_i).
\end{equation}
Based on the research on linear approximation conditions for nonlinear models
\cite{1980Relative},
the constraint is a linear approximation
at the stable operating point of the CAO-SIR-FBC method.
In other words, the constraint is simplified to
\begin{equation}
	1-\sum_{i=1}^M\varepsilon_i\geq 1-\varepsilon_d '\ ,
\end{equation}
where $\varepsilon_d '$ is a parameter in the neighborhood of
$\varepsilon_d$.
As a result, the inequality constraints of the optimization problem (\ref{opt: original})
is expressed as
\begin{equation}\label{opt: varepsilon_d}
			\sum\limits_{k=1}^N\varepsilon_D[k]+\sum\limits_{k=1}^N\sum\limits_{m=k}^{N}\varepsilon_m[k]\leq \varepsilon_d ' ,
\end{equation}
where the range of $\varepsilon_d '$ is given by the following lemma.
\begin{lemma} \label{LEMMA 1}
The range of $\varepsilon_d '$ is given by
\begin{equation}
\begin{split}
    \varepsilon_{\min}\leq& \varepsilon_d '\leq \varepsilon_{\max}
		,
\end{split}
\end{equation}
where $\varepsilon_{\min}=\varepsilon_d$ and $\varepsilon_{\max}=\frac{M-\sqrt{M^2-2M(M-1)\varepsilon_d}}{M-1}$.
\end{lemma}
\begin{proof}
See Appendix \ref{sec:appendices 1}.
\end{proof}
Here, $\varepsilon_{\min}$ and $\varepsilon_{\max}$ denote the infimum and supremum of $\varepsilon_d'$, respectively.

Next, we define a kind of parameter $x$ that satisfies
\begin{equation}\label{eq: define x}
	x := Q^{-1}(\varepsilon).
\end{equation}
By substituting Eq. (\ref{eq: define x}) into
Eqs. (\ref{eq:epsilon_D[k]}) and (\ref{eq:epsilon_m[k]}), we have
a simple linear relation of parameter $x$ .
Consequently, after substituting parameter $x$ for $\varepsilon$, the optimization problem (\ref{opt: varepsilon_d}) is simplified as
\begin{equation}\label{opt: constrained x}
	\begin{split}
		\min \quad &\sum\limits_{k=1}^N \sqrt{V_{S, k}}
		\quad x_k[k]\\
		s.t. \quad &\begin{cases}
			\sum\limits_{k=1}^N\left(\sum\limits_{m=k}^{N}Q(x_m[k])+Q(x_D[k])\right)\leq \varepsilon_d ' \\
			x_m[k]=\sqrt{\frac{n}{V_{S, m}}}(C_{S, m}-C_{S, k})
			+\sqrt{\frac{V_{S, m}}{V_{S, k}}}x_k[k]                                                    \\
			x_D[k]=\sqrt{\frac{n}{V_{k, D}}}(C_{k, D}-C_{S, k})+
			\sqrt{\frac{V_{S, k}}{V_{k, D}}}x_k[k]                                                     \\
			x_k[k]>0,  k=1,2,\cdots,N,
		\end{cases}
	\end{split}
\end{equation}
where the constraint of $x_k[k]$ is attributed to $\varepsilon_k[k]<\varepsilon_d'<0.5, k=1,2,\cdots,N$.

The optimization problem (\ref{opt: constrained x}) is actually a convex optimization. However, the gradient of $Q(x)$ at the operating point satisfies $\nabla Q(x) \ll 1$. Hence, the convergence is very slow. It is complicated to obtain the solution through
the traditional interior-point method \cite{2004Convex} or CVX toolkit \cite{2008CVX}.
A further simplification is necessary.

Based on the optimization problem (\ref{opt: constrained x}), we propose the following lemma and substitute the
inequality constraints with an equality constraint.
\begin{lemma}
	\label{LEMMA 2}
The solution to the optimization problem (\ref{opt: constrained x}) satisfies
\begin{equation}
		\sum\limits_{k=1}^N\left(\sum\limits_{m=k}^{N}Q(x_m[k])+Q(x_D[k])\right)= \varepsilon_d ' \ .
	\end{equation}
\end{lemma}
\begin{proof}
See Appendix \ref{sec:appendices 2}.
\end{proof}
We then transform the optimization problem (\ref{opt: constrained x}) into an abstract optimization problem to analyze its properties.
The coefficient matrices and coefficient vector of the optimization problem (\ref{opt: constrained x}) are denoted by $(a_{ij})_{N\times (N+1)}, (b_{ij})_{N\times (N+1)}$ and $(c_i)_{N\times 1}$,
respectively, in which the matrix element is denoted by
\begin{equation*}\label{pramater substitute}
	\begin{split}
		&a_{ij}=\sqrt{\frac{V_{S, i}}{V_{S, j}}}, \\
		&b_{ij}=\sqrt{\frac{n}{V_{S, j}}}(C_{S, j}-C_{S, i}), \qquad  1\leq j\leq N, \\
		&a_{i(N+1)}=\sqrt{\frac{V_{S, i}}{V_{i, D}}}\ , \\
		&b_{i(N+1)}=\sqrt{\frac{n}{V_{i, D}}}(C_{i, D}-C_{S, i}), \\
		&c_i=\sqrt{V_{S, i}}.
	\end{split}
\end{equation*}
Here, we regard the destination node $D$ as an $N+1$ relay because of the similarity between its rate equation and that of other relays. Furthermore, an optimization variable $\boldsymbol{x}=(x_1, x_2, \cdots, x_N)$ is adapted, in which the element $x_k$ denotes $x_k[k]$ for $k=1, 2, \cdots, N$.

Substituting coefficient matrices and coefficient vectors into the optimization problem (\ref{opt: constrained x}), we obtain the optimization problem expressed as
\begin{equation}\label{eq:simplest  constraint}
	\begin{split}
		\min \quad &f(\boldsymbol{x})=\sum\limits_{i=1}^N c_{i}x_i\\
		s.t.\quad &\begin{cases}
			 & h(\boldsymbol{x})=\sum\limits_{i=1}^N \sum\limits_{j=i}^{N+1}
			Q(a_{i, j}x_{i}+b_{i, j})-\varepsilon_d '=0                      \\
			 & x_{i}\geq 0,  i=1,2,\cdots,N .
		\end{cases}
	\end{split}
\end{equation}

It is worth noting that the optimization problem (\ref{eq:simplest  constraint}) is not a convex optimization because the equality constraint function $h(\boldsymbol{x})$ is not affine.
As a result, we propose Lemma \ref{LEMMA 3} to exchange the equality constraint function and objective functions.
\begin{lemma} \label{LEMMA 3}
Let the solution of the optimization problem (\ref{eq:simplest  constraint}) be denoted by $\boldsymbol{x_0}$. The value of the objective function is presented by $v_0$.
Then, the optimization problem with the same solution $\boldsymbol{x_0}$ of the optimization problem (\ref{eq:simplest  constraint}) is expressed as
\begin{equation}\label{eq:final constraint}
		\begin{split}
			\min \quad &f(\boldsymbol{x})=\sum\limits_{i=1}^N \sum\limits_{j=i}^{N+1}
			Q(a_{i, j}x_{i}+b_{i, j})\\
			s.t.\quad &\begin{cases}
				 & h(\boldsymbol{x})=\sum\limits_{i=1}^N c_{i}x_i-v_0=0 \\
				 & x_{i}\geq 0, i=1,2,\cdots,N.
			\end{cases}
		\end{split}
	\end{equation}
\end{lemma}
\begin{proof}
See Appendix \ref{sec:appendix 3}.
\end{proof}
In the optimization problem (\ref{eq:final constraint}), the equality constraint function is affine, and the objective function
is convex because
\begin{equation*}
	\nabla^2 f(\boldsymbol{x})\geq 0,
\end{equation*}
for all $x_i>0$.
Consequently, problem (\ref{eq:final constraint}) is a convex optimization problem. The minimal $f(\boldsymbol{x})$ is obtained with a specific $v_0$. Additionally, the range of $v_0$ can be obtained by the following lemma.
\begin{lemma}\label{LEMMA 4}
    The objective function value in the solution of the optimization problem (\ref{eq:simplest  constraint}) satisfies
    \begin{align}
    &v_1 \leq v_0 \leq v_2,
\end{align}
where $v_1=Q^{-1}\left(\varepsilon_d'\right)\sum\limits_{i=1}^N c_i$ and $v_2=x_1\sum\limits_{i=1}^N c_i$.
Here, the parameter $ x_1$ is the solution to the equation expressed as
\begin{equation}
	\sum\limits_{i=1}^N \sum\limits_{j=i}^{N+1}
	Q(a_{i, j}x+b_{i, j})-\varepsilon_d '=0.
\end{equation}
\begin{proof}
    The vector $\boldsymbol{x_1}$ is equal to $x_1\cdot (1)_{1\times N}^T$ and satisfies the constraint of problem (\ref{eq:simplest  constraint}), so we have \begin{equation}
    v_0\leq v_2.
\end{equation}
Furthermore, because every error probability satisfies $\varepsilon_m[k]\leq\varepsilon_d'$, the solution $v_0$ satisfies
\begin{equation}
    v_0 \geq v_1.
\end{equation}
\end{proof}
\end{lemma}
\begin{figure}[t]
	\centering
	\includegraphics[width=0.45\textwidth]{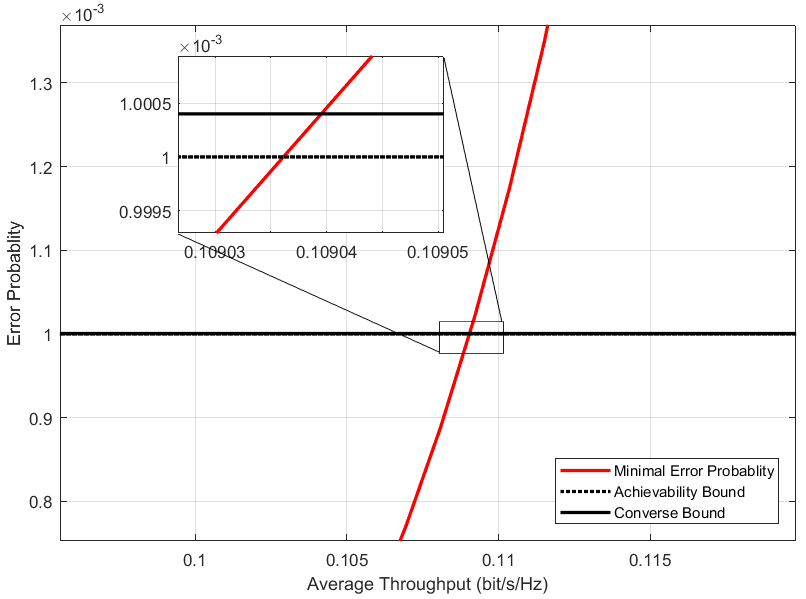}
\caption{Computing process on $\bar{R}^*$ for $i.i.d.$ cases with relay number $N=2$, blocklength $n=300$, transmission power $P=1$, and error probabilities $\varepsilon_d=10^{-3}$. The channel gain obeys an index distribution, and $\bar{g}_{a, b}=1$.}
	\label{fig:scan_rate}
\end{figure}
Accordingly, the solution to the optimization problem (\ref{eq:simplest  constraint}) is obtained using the bisection algorithm. More specifically, the midpoint of the interval $v_3=\frac{v_1+v_2}{2}$ is substituted into the optimization problem (\ref{eq:final constraint}).
The solution $\boldsymbol{x_3}$ of the optimization problem (\ref{eq:final constraint}) with $v_3$ is calculated. Then, the objective function $f(\boldsymbol{x_3})$ is compared with $\varepsilon_d$. If $f(\boldsymbol{x_3})-\varepsilon_d$ is sufficiently small or $\frac{v_2-v_1}{2}$ is smaller than the tolerance, the convergence is satisfactory. However, $v_1$ is replaced with $v_3$ if $f(\boldsymbol{x_3})<\varepsilon_d$. When $f(\boldsymbol{x_3})>\varepsilon_d$, $v_3$ replaces $v_2$ so that there is a solution crossing within the new interval. A detailed process is shown in Algorithm \ref{al:scan rate}.

The complexity of Algorithm \ref{al:scan rate} is $O\left(\log_2\left({Q^{-1}\left(\varepsilon_d\right)}\right)\right)$. In contrast, the complexity of the exhaustive method is $O\left({Q^{-1}\left(\varepsilon_d\right)}\right)$. The proposed method significantly reduces the computational
complexity.

\begin{algorithm}[t]
	\caption{Compute the maximal average rate  $\bar{R}^*$ with specific $\varepsilon_d'$}\label{al:scan rate}
Initialize pre-existing parameter $P, n$, interval $[v_1, v_2]$, iteration number $k$,  etc\\
Compute Shannon capacity $C_{a, b}$, channel dispersion $V_{a, b}$, parameter matrices $(a_{ij})_{N\times (N+1)}, (b_{ij})_{N\times (N+1)}$, coefficient vector $(c_i)_{N\times 1}$\\
\While{$k<k_{\max}$}
{
Compute solution $\boldsymbol{x_0}$ to optimization problem (\ref{eq:final constraint})  through the interior-point method\\
Compute objective functions $f\left(\boldsymbol{x_0}\right)$\\
\If{$\left\lvert f(\boldsymbol{x_0})-\varepsilon_d ' \right\rvert <\varepsilon$ $\mathbf{or}$ $\frac{v_2-v_1}{2}<TOL$}
{
Compute and return  $\bar{R}^*$
}
\eIf{$f(\boldsymbol{x_0})<\varepsilon_d '$}
{$v_1\leftarrow v_3$}
{$v_2\leftarrow v_3$}
}
{Return message `` $\bar{R}^*$ is not found''}
\end{algorithm}

Based on Algorithm \ref{al:scan rate}, we propose a method of computing the approximated maximal CAO-SIR-FBC rate denoted by $\hat{R}$. As shown in Fig. \ref{fig:scan_rate}, the difference in probabilities between the achievability bound
and the converse bound is slight in the general case $(N \leq 10)$. Consequently, a simple approximated maximal average rate $\hat{R}$ is obtained by
\begin{equation}
	\hat{R}=f(\hat{\varepsilon}_d ),
\end{equation}
where the approximated error probability $\hat{\varepsilon}_d$  is given by
\begin{equation}
	\hat{\varepsilon}_d=\frac{1}{2}\left(\varepsilon_d+\frac{M-\sqrt{M^2-2M(M-1)\varepsilon_d}}{M-1}\right).
\end{equation}
In most cases, the approximated rate is reliable.
Furthermore, if the percent error of rate $\hat{R}$ is too high,
we examine every $\varepsilon_d '$ in the range [$\varepsilon_{\min}, \varepsilon_{\max}$] with a step $\Delta \varepsilon$. Then, we compute and cache every $\boldsymbol{x_0}$ of $\varepsilon_d '$. With $\boldsymbol{x_0}$, $\zeta$ is obtained using Eq. (\ref{eq: lower limit of success probability}). A more accurate approximated solution is obtained when the difference is less than the error bound.
The specific process is shown in Algorithm \ref{al:compute approximated rate}.
\begin{algorithm}[t]
\caption{Compute the approximated maximal $\hat{R}$ }\label{al:compute approximated rate}
Initialize the pre-existing parameter of Algorithm \ref{al:scan rate}, error bound
$\varepsilon$,  etc\\
Compute $\varepsilon_{\min}$, $\varepsilon_{\max}$, approximated error probability $\hat{\varepsilon_d }$, achievability bound $R_a$, converse bound
$R_c$, the approximated value $\hat{R}$ and the percent error of $\hat{R}$ are expressed as
\begin{equation}
		\delta =\frac{\max\{ R_c-\hat{R}, \hat{R}-R_a  \}}{\hat{R}}
	\end{equation}
\eIf{$\delta<\alpha$}
{Return $\hat{R}$}
{
Compute $\Delta \varepsilon=(\varepsilon_{\max}-\varepsilon_{\min})/{k}$\\
\For{$\varepsilon_d '=\varepsilon_{\min}:\Delta \varepsilon:\varepsilon_{\max}$}
{
Compute solution $\{R_0, \boldsymbol{x_0}\}$ by Algorithm \ref{al:scan rate}\\
Compute $\zeta$ of $\boldsymbol{x_0}$ \\
\If{$\left\lvert \varepsilon_d-\left(1-\zeta\right)\right\rvert<\varepsilon$ }
{Return $\hat{R}=R_0$}
}
}
\end{algorithm}
Now, we obtain a specific method to compute a reliable $\hat{R}$. The analysis of the protocol performance is
presented in Section \ref{sec:PERFORMANCE ANALYSIS}.

\section{Performance Analysis}\label{sec:PERFORMANCE ANALYSIS}
In this section, the performance of the CAO-SIR-FBC
protocol is analyzed. We are interested in the average throughput
in a fading channel. In this section, the channel spanning between node $a$ and node $b$ is assumed to be an independent distributed Rayleigh fading channel. The probability density function $(p.d.f)$ of the channel gain $g_{a,b}$ is obtained from $f_{g_{a,b}}(x)=\frac{1}{\bar{g}_{a,b}}\ \exp(-\frac{x}{\bar{g}_{a,b}})$. Here, $\bar{g}_{a,b}=\mathbb{E}g_{a,b}=\mathbb{E}|h_{a,b}|^2$ denotes the average channel gain of the link from node $a$ to $b$.

Because it is not trivial to derive a closed-form expression of the average throughput of the CAO-SIR-FBC method, we
focus on the approximation of the average throughout in the high blocklength, low and high SNR.
Consider the high blocklength first. Because the rate $\bar{R}^*$ increases and trends to its capacity when $n$ increases, we have
\begin{equation}
	\lim_{n \to \infty}  \bar{R}^*(\varepsilon, n)=\bar{C}.
\end{equation}
Here, the capacity is presented by \cite{2014CAO}
\begin{align}
	 & \bar{C}=\frac{1}{N+1}\sum\limits_{i=1}^N C_i,
\end{align}
where $ C_i=\min\left\{ C_{S, i}, C_{i, D} \right\}$.
In other words, when the blocklength of coding $n$ tends to infinity, the CAO-SIR-FBC protocol is equivalent to the basic CAO-SIR. As a result, the average throughput is given by
\begin{equation}
	\begin{split}
		\lim_{n \to \infty} \mathbb{E}\{\bar{R}^*(\varepsilon, n)\}
		 =&\mathbb{E}\{\bar{C}\}\\
		=&\frac{2N}{N+1}\int\limits_0^\infty \log(1+x\gamma)\exp(-2x)\, \mathrm{d}x .
	\end{split}
\end{equation}

Next, we focus on the CAO-SIR-FBC protocol's performance in the
SNR domain.
In a low-SNR regime, due to channel dispersion, the
CAO-SIR-FBC rate is 0, as shown in Fig. \ref{fig:rate_p}. In other words, it is impossible to transmit the message with the desired error probability $\varepsilon_d$.
\begin{figure}[t]
	\centering
	\includegraphics[width=0.45\textwidth]{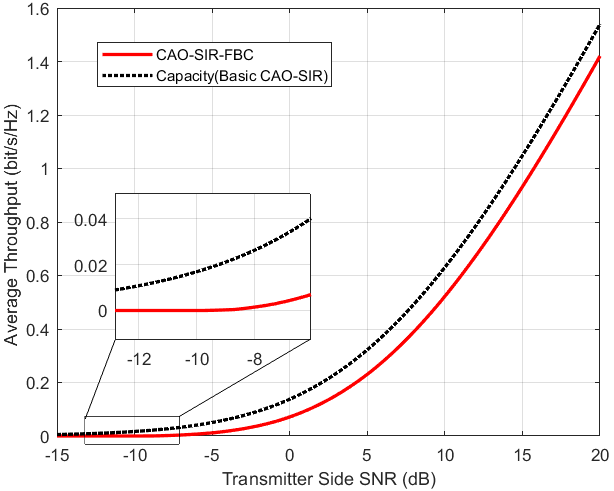}
\caption{Maximal average rate $\bar{R}^*$ for $i.i.d.$ cases with relay number $N=2$, blocklength $n=300$, and error probabilities $\varepsilon_d=10^{-3}$. The channel gain obeys an index distribution, and $\bar{g}_{a, b}=1$.}
	\label{fig:rate_p}
\end{figure}
In a high SNR regime, it is worth noting that $V$ trends to a constant presented by
\begin{equation}
	\lim\limits_{\rho  \to \infty} V\left(\rho\right)=\frac{(\log \mathrm{e})^2}{2}.
\end{equation}
Accordingly, the $h(\boldsymbol{x})$ in the high SNR regime is obtained by
\begin{equation}\label{eq: constraint in high SNR}
	\begin{split}
			      h(\boldsymbol{x})&=\sum\limits_{i=1}^N \sum\limits_{j=i}^{N+1}
			Q\left(x_{i}+\frac{\sqrt{2n}}{\log \mathrm{e}} \log\left(\frac{g_{S, j}}{g_{S, i}}\right)\right)\\
			&\qquad-\varepsilon_d '\\
			&=0
	\end{split}
\end{equation}
It is not difficult to verify that the optimization problem  in the high SNR regime is independent of an arbitrary SNR. Its solution  is denoted by a constant vector $\boldsymbol{x_0}=\big[x_0\left[1\right], x_0[2], \cdots, x_0\left[N\right]\big]$, and we have
\begin{equation}
	f(\boldsymbol{x_0})=\sqrt{\frac{(\log \mathrm{e})^2}{2}}\sum\limits_{i=1}^N x_0[i].
\end{equation}
As a result, the average throughput of the CAO-SIR-FBC protocol in a high SNR regime is given by
\begin{equation}\label{eq: R in high SNR}
	\begin{split}
		\mathbb{E}\{\bar{R}^*\}
		=&\mathbb{E}\{\bar{C}\}-\frac{1}{N+1}
		\left(\sqrt{\frac{(\log \mathrm{e})^2}{2n}}\sum\limits_{i=1}^N x_0[i]-N\frac{\log n}{2n}\right)\\
		=&\frac{2N}{N+1}\int\limits_0^\infty \log(1+x\gamma)\exp(-2x)\, \mathrm{d}x\\
		&-\frac{1}{N+1}\left(\sqrt{\frac{(\log \mathrm{e})^2}{2n}}\sum\limits_{i=1}^N x_0[i]-N\frac{\log n}{2n}\right).
	\end{split}
\end{equation}
As shown in Eq. (\ref{eq: R in high SNR}), the difference between
$R^*$ and $C$ trends to a constant in a high SNR regime.
In addition, the slope of the CAO-SIR-FBC method is equal to its capacity in a high SNR regime. Channel dispersion does not influence the slope, also referred to as multiplexing gain, of the CAO-SIR-FBC protocol.

\section{Numerical Results}\label{ sec:NUMERICAL RESULTS}

Numerical results are presented in this section to verify the above analysis and the advantage of the CAO-SIR-FBC method. Assume that
there are three potential relays. Both independent and
identically distributed ($i.i.d.$) and independent nonidentical distributed ($i.ni.d.$) Rayleigh channels \cite{2006Fundamentals, 2013Outage} are considered. For the $i.i.d.$ case, we suppose the channel gains $\frac{1}{\bar{g}_{a, b}}=1$ for every node pair $(a, b)$. For the $i.ni.d.$ case, we
assume $\frac{1}{\bar{g}_{a, b}}=d_{a, b}^{-2}$, where $d_{a, b}$ denotes the distance between node pair $(a, b)$, because of the pass-loss model. To determine $d_{a, b}$, we assume the following network topology. Nodes $S$ and $D$ are located in coordinates (0, 0) and (0, 1), respectively. The coordinates of the three relays are given by ($\frac{\sqrt{3}}{2}, \frac{1}{2}$),
($-\frac{\sqrt{3}}{2}, \frac{1}{2}$), and ($0, \frac{1}{2}$). To provide more
insight, the CAO-SIR-FBC protocols are compared with the conventional two-timeslot
DF relay protocols (with and without the direct $S-D$ link) with a fixed target rate \cite{laneman2003distributed, laneman2004cooperative, bletsas2006simple} and two-timeslot DF relaying without the direct $S-D$ link \cite{bletsas2006simple}. The average throughput of traditional protocols also considers the influence of finite-blocklength coding. The
fixed target rate is $r=1$. Different from Eq. (\ref{eq:C}), the transmission rate is computed in the complex dimension to easily analyze the slope and multiplexing gains. For other system parameters, blocklength $n=300$, and the error probability at the destination node is $\varepsilon_d=10^{-3}$.

\begin{figure}[t]
	\centering
	\includegraphics[width=0.45\textwidth]{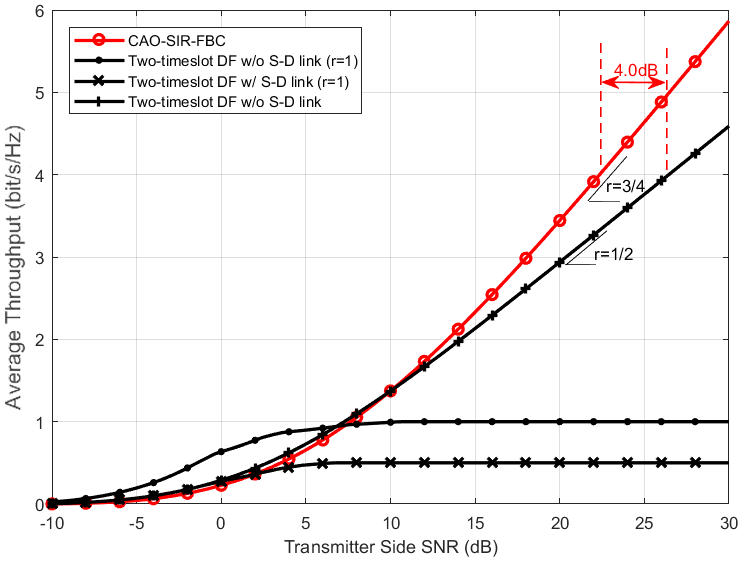}
\caption{
Average throughput in the $i.i.d.$ case. The proposed CAO-SIR-FBC protocols are compared with conventional two-timeslot relaying.
}
	\label{fig:iid}
\end{figure}
\begin{figure}[t]
	\centering
	\includegraphics[width=0.45\textwidth]{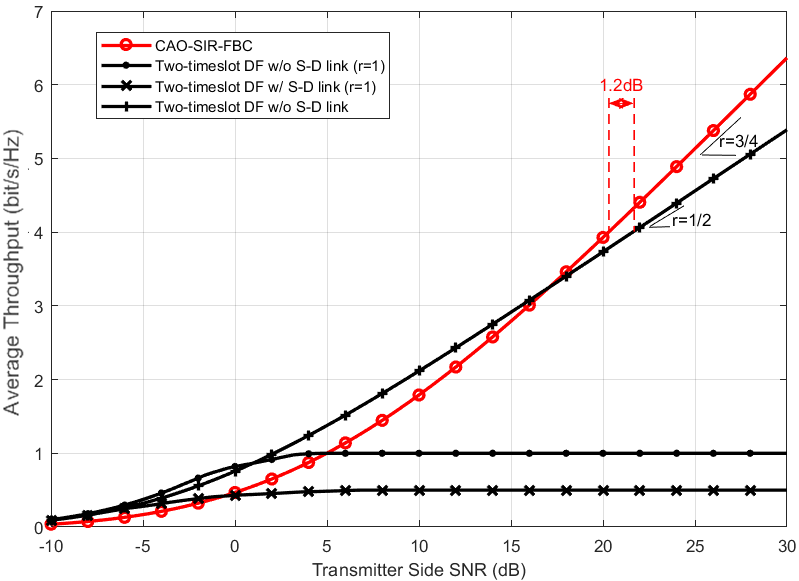}
\caption{
Average throughput in the $i.ni.d.$ case. The proposed CAO-SIR-FBC protocols are compared with conventional two-timeslot relaying.
}
	\label{fig:inid}
\end{figure}
The average throughput versus transmitter-side SNR curves
for $i.i.d.$ and $i.ni.d.$ are presented in Figs. \ref{fig:iid} and \ref{fig:inid}, respectively. In the high SNR regime, we focus on the slope of the throughput, which also means the multiplexing gains. The multiplexing gain of the CAO-SIR-FBC method in both the $i.i.d.$ and $i.ni.d.$ cases is
$\frac{3}{4}$. The multiplexing gain of the CAO-SIR-FBC method is equal to that of the basic CAO-SIR \cite{2014CAO}, which verifies the theory that the channel dispersion does not influence the multiplexing gain. The multiplexing gain of two-timeslot DF relaying without the direct $S-D$ link is $\frac{1}{2}$ in both the $i.i.d.$ and $i.ni.d.$ cases, which is much lower than that of CAO-SIR-FBC because it requires two timeslots to transmit one message. Then, we focus on the required SNR when the average throughput is $4$ bits/s/Hz. In the $i.i.d.$ case, the CAO-SIR-FBC
method has SNR gains of $4.0$ dB over two-timeslot DF relaying without the direct $S-D$ link. For the $i.ni.d.$ scenario, the SNR gains of CAO-SIR-FBC over two-timeslot DF relaying without the direct $S-D$ link are $1.2$ dB. This shows that in the high SNR regime, the employed relays of CAO-SIR-FBC are beneficial to
obtain a high average throughput.

Next, we focus on comparing CAO-SIR-FBC and conventional relaying protocols with fixed target rates. It should be noted that it might be unfair to
set the fixed target rate of conventional relaying protocols,
because the average throughput for fixed-rate relaying is constrained in the high SNR regime. Specifically, we focus on a transmitter-side SNR of 20 dB. With the fixed target rate $r=1$, the average throughput of the two-timeslot
DF relaying with and without direct $S-D$ links is $1$ and $\frac{1}{2}$ bit/s/Hz, respectively, regardless of whether the  case is $i.i.d.$ or $i.ni.d.$.
Although the upper bound in the high SNR regime is increased with the fixed rate, it is worth noting that a higher target rate also causes a lower average throughput in the low-SNR regime. In addition, even without a fixed target rate, two-timeslot DF relaying without the $S-D$ link is still only $2.9 $ bit/s/Hz in the $i.i.d.$ case and $3.7$ bit/s/Hz in the $i.ni.d.$ case. In contrast, because the CAO-SIR-FBC is capable of adapting its transmission rate to the channel gains, it always achieves the optimal average throughput in the high SNR regime. In the $i.i.d.$ case, CAO-SIR-FBC achieves throughput gains of $117\%$, $344\%$, and $688\%$ over two-timeslot DF relaying and fixed-rate two-timeslot DF relaying with and without direct $S-D$ links, respectively. In the $i.ni.d.$ case, the throughput gains of CAO-SIR-FBC over two-timeslot DF relaying as well as fixed-rate two-timeslot DF relaying with and without direct $S-D$ links are $105\%$, $392\%$, and $784\%$, respectively.
Hence, we can conclude that the joint channel-aware rate-adaptation mechanism greatly improves the efficiency of successive relaying.

\section{Conclusions}\label{sec:Conclusions }
This paper presents an optimized relaying protocol referred to as CAO-SIR-FBC based on the CAO-SIR scheme.
Faced with a short-packet transmission in which the Shannon formula is not achievable, we analyze the optimization problems of a successive relaying protocol with FBC in detail.
We then solve the optimization problem and clarify the adjustments of the successive relaying protocol with the FBC.
To analyze the performance of the CAO-SIR-FBC protocol, we carried out a series of simulation experiments and analyzes. Experimental results show that the CAO-SIR-FBC protocol has high reliability and excellent transmission performance in short-packet transmission, and its performance is better than those of conventional protocols with FBC.
In the future, we hope to continue studying the relay ranking, relay selection, power allocation, and other mechanisms of the CAO-SIR-FBC protocol for improvement.


%

\appendices
\section{PROOF OF THEOREM \ref{THEOREM 1}\label{sec:appendices 4}}
The proof is based on the following observation. In CAO-SIR-FBC, the relay's protocol is successive decoding. More specifically, $W_1$ is decoded by each relay in the first timeslot. In the second timeslot, $W_2$ is decoded successfully only if $W_1$ is decoded successfully. Likewise, message $W_k$ is decoded successfully only if
messages $W_1, W_2, \cdots W_{k-1}$ are decoded successfully. As a result,
error probabilities are increased in each timeslot, also referred to as error propagation.

Let $\hat{W}_k$ denote the message decoded by relay $k$ in the $k$th timeslot. The message $W_k$ is decoded successfully by relay $k$ only if $\hat{W}_k=W_k$.
Consequently, for $2\leq  k\leq  N$, the probability of decoding success, denoted by $\mathrm{Pr}\left\{\hat{W}_k=W_k\right\}$, is presented by

\begin{equation}
	\begin{split}
		\mathrm{Pr}\left\{\hat{W}_k=W_k\right\}		=&\mathrm{Pr}\left\{\hat{W}_{1}=W_{1}\right\}\\
		&\times\prod\limits_{i=1}^{k-1}\mathrm{Pr}\left\{\hat{W}_{i+1}=W_{i+1}|\hat{W_{i}}=W_{i}\right\}.
	\end{split}
\end{equation}
Additionally, based on the definition of $\varepsilon$, when $k=1$, we have
\begin{equation}
	\mathrm{Pr}\left\{\hat{W}_{1}=W_{1}\right\} \geq 1-\varepsilon_1[1].
\end{equation}

When $2\leq  k\leq  N$, because the event $\left\{\hat{W}_k=W_k|\hat{W}_{k-1}=W_{k-1}\right\}$ means that relay $k$ decodes the message $W_1, W_2, \cdots W_{k}$ successfully,
it is obtained that
\begin{equation}
	\mathrm{Pr}\left\{\hat{W}_k=W_k|\hat{W}_{k-1}=W_{k-1}\right\}\geq \prod_{i=1}^{k}(1-\varepsilon_k[i]).
\end{equation}
Based on the multiplication theorem, the probability of relay $N$ decoding message
$W_N$ successfully satisfies
\begin{equation}\label{eq:P(A) min}
	\mathrm{Pr}\left\{\hat{W}_{N}=W_{N}\right\}\geq \prod_{k=1}^{N}\prod_{i=1}^{k}
	(1-\varepsilon_k[i]).
\end{equation}
Rewritten according to relay number $m$, Eq. (\ref{eq:P(A) min}) is expressed as
\begin{equation}\label{eq:P(A) min2}
	\mathrm{Pr}\left\{\hat{W}_{N}=W_{N}\right\}\geq \prod_{k=1}^{N}\prod_{m=k}^{N}
	(1-\varepsilon_m[k]).
\end{equation}
For node $D$, the protocol is reverse successive decoding.
Specifically, $W_k$ is decoded successfully at node $D$ only if
$W_{k+1}, W_{k+2}, \cdots, W_N$ are also decoded successfully.
Let $\hat{W}_k^D$ denote the estimate of $W_k$ at the destination node $D$. The message $W_k$ is decoded successfully by $D$ only if $\hat{W}_k^D=W_k$.
For $1\leq  k\leq  N-1$, the probability of decoding success at node $D$, denoted by $\mathrm{Pr}\left\{\hat{W}_k^D=W_k\right\}$, is expressed as
\begin{equation}
	\begin{split}
		\mathrm{Pr}\left\{\hat{W}_k^D=W_k\right\}		=&\mathrm{Pr}\left\{\hat{W}_{N}^D=W_{N}\right\}\\
		&\times\prod\limits_{i=k}^{N-1}\mathrm{Pr}\left\{\hat{W}_{i}^D=W_{i}|\hat{W}_{i+1}^D=W_{i+1}\right\}.
	\end{split}
\end{equation}
When $k=N$,
message $W_N$ is decoded
correctly at $D$ only if relay $N$ decodes the message successfully, so $\mathrm{Pr}\left\{\hat{W}_{N}^D=W_{N}\right\}$ is given by
\begin{equation}
	\begin{split}
		\mathrm{Pr}\left\{\hat{W}_{N}^D=W_{N}\right\}&=\mathrm{Pr}\left\{\hat{W}_{N}^D=W_{N}|\hat{W}_{N}=W_{N}\right\}\\
		&\qquad\times\mathrm{Pr}\left\{\hat{W}_{N}=W_{N}\right\}.
	\end{split}
\end{equation}
Hence, $\mathrm{Pr}\left\{\hat{W}_{N}^D=W_{N}\right\}$
satisfies
\begin{equation}
	\begin{split}
		\mathrm{Pr}\left\{\hat{W}_{N}^D=W_{N}\right\}
		\geq \mathrm{Pr}\left\{\hat{W}_{N}=W_{N}\right\}(1-\varepsilon_D[N]).
	\end{split}
\end{equation}
Additionally, when $1\leq  k\leq  N-1$, we have
\begin{equation}
	\mathrm{Pr}\left\{\hat{W}_{k}^D=W_{k}|\hat{W}_{k+1}^D=W_{k+1}\right\}\geq 1-\varepsilon_D[k].
\end{equation}
As a result, the minimal probability of $\left\{\hat{W}^D_{1}=W_{1}\right\}$ satisfies
\begin{equation}\label{eq:all error probability}
	\mathrm{Pr}\left\{\hat{W}^D_{1}=W_{1}\right\} \geq \prod_{k=1}^{N}(1-\varepsilon_D[k])
	\prod_{k=1}^{N}\prod_{m=k}^{N}
	(1-\varepsilon_m[k]).
\end{equation}
We use $\zeta$ to denote the right-hand side of Eq. (\ref{eq:all error probability}), which is also the lower limit of $\mathrm{Pr}\left\{\hat{W}^D_{1}=W_{1}\right\}$, so we have
\begin{equation} \label{eq: lower limit of success probability1}
    \zeta=\prod_{k=1}^{N}(1-\varepsilon_D[k])
	\prod_{k=1}^{N}\prod_{m=k}^{N}
	(1-\varepsilon_m[k]).
\end{equation}
It is worth noting that $\mathrm{Pr}\left\{\hat{W}^D_{1}=W_{1}\right\}$ is also the probability of $D$ decoding all messages without errors, which is referred to as the reliability of the communication system. Therefore, the theorem \ref{THEOREM 1} is proved.

\section{PROOF OF LEMMA \ref{LEMMA 1}}\label{sec:appendices 1}
The proof is based on the following observation.
The expansion of Eq. (\ref{eq: lower limit of success probability})
with the first and second remainder terms expressed as
\begin{equation}
		\begin{split}
			\zeta=&1-\sum_{i=1}^{M}\varepsilon_i+o\left(\sum_{i=1}^{M}\varepsilon_i\right)\\
			=&1-\sum_{i=1}^M\varepsilon_i+\sum_{i=1}^{M-1}\sum_{j=i+1}^M \varepsilon_i\varepsilon_j
			-o\left(\sum_{i=1}^{M-1}\sum_{j=i+1}^M \varepsilon_i\varepsilon_j\right),
		\end{split}
	\end{equation}
where the remainder term $o\left(\sum_{i=1}^{M}\varepsilon_i\right)$
means that it is the
higher order infinitesimal of $\sum_{i=1}^{M}\varepsilon_i$.
Similarly, the remainder term $o\left(\sum_{i=1}^{M-1}\sum_{j=i+1}^M \varepsilon_i\varepsilon_j\right)$ is
the
higher order infinitesimal of $\sum_{i=1}^{M-1}\sum_{j=i+1}^M \varepsilon_i\varepsilon_j$.
Because each remainder term is positive definite,
we have the following inequality
\begin{align}
		\zeta & >1-\sum_{i=1}^{M}\varepsilon_i, \label{ineq: B bigger}                                  \\
		\zeta & <1-\sum_{i=1}^M\varepsilon_i+\sum_{i=1}^{M-1}\sum_{j=i+1}^M \varepsilon_i\varepsilon_j.
		\label{ineq: B small}
	\end{align}
For the second-order term on the right-hand side of Eq. (\ref{ineq: B small}), because of the inequality of arithmetic and geometric means (AM-GM inequality),
we have
\begin{equation}\label{inqe:2nd term}
		\sum_{i=1}^{M-1}\sum_{j=i+1}^M \varepsilon_i\varepsilon_j
		\leq \frac{M-1}{2}\sum_{i=1}^M \varepsilon_i^2
		<\frac{M-1}{2}\varepsilon_d\sum_{i=1}^{M}\varepsilon_i.
	\end{equation}
Here, the first term is equal to the second term if and only if every $\varepsilon_i$ is equal to each other for $i=1, 2, \cdots, M$.
The second inequality relation results from
$\varepsilon_i<\varepsilon_d$ for $i=1, 2, \cdots, M$.

Additionally, we substitute Eq. (\ref{inqe:2nd term}) into
the expansion of $(\sum_{i=1}^M \varepsilon_i)^2$. The intersection
of inequality is expressed as
\begin{equation}
		\begin{split}
			\left(\sum_{i=1}^M \varepsilon_i\right)^2&=
			\sum_{i=1}^M \varepsilon_i^2
			+2\sum_{i=1}^{M-1}\sum_{j=i+1}^M \varepsilon_i\varepsilon_j\\
			&\geq \frac{2M}{M-1}
			\sum_{i=1}^{M-1}\sum_{j=i+1}^M \varepsilon_i\varepsilon_j.
		\end{split}
	\end{equation}
Thus, we have
\begin{equation}
		\sum_{i=1}^{M-1}\sum_{j=i+1}^M \varepsilon_i\varepsilon_j
		\leq \frac{M-1}{2M}\left(\sum_{i=1}^M \varepsilon_i\right)^2.
		\label{inqe:2nd term 2}
	\end{equation}
By substituting the inequality
(\ref{ineq: B small}) into the right-hand side of inequality (\ref{inqe:2nd term 2})
and the third term of inequalities (\ref{inqe:2nd term}), it follows that
\begin{equation}
		\begin{cases}
			\zeta<
			1-\frac{2-(M-1)\varepsilon_d}{2}\sum\limits_{i=1}^M\varepsilon_i, \\
			\zeta<1-\sum\limits_{i=1}^M\varepsilon_i
			+\frac{M-1}{2M}\left(\sum\limits_{i=1}^M\varepsilon_i\right)^2.
		\end{cases}
	\end{equation}
By substituting
the right-hand side of inequality (\ref{ineq: B bigger}) for $\zeta$ in inequality (\ref{ineq: constraint}),
we have
\begin{equation}\label{ineq:lower bound substitution inequality}
		1-\sum_{i=1}^M\varepsilon_i\geq 1-\varepsilon_d'.
	\end{equation}
It is important to note that Eq. (\ref{ineq:lower bound substitution inequality}) is the equivalent condition of Eq. (\ref{ineq: constraint}). According to Eq. (\ref{ineq: B bigger}), the right-hand side of Eq. (\ref{ineq: constraint}) should be greater than the right-hand side of Eq. (\ref{ineq:lower bound substitution inequality}), so $\varepsilon_d '$
satisfies
\begin{equation}
		\varepsilon_d'\geq \varepsilon_d.
	\end{equation}
Likewise, we obtain two converse bounds of
$\bar{R}^*$ by substituting $\zeta$
into the right-hand side of inequalities (\ref{inqe:2nd term})
and (\ref{inqe:2nd term 2}).
The inequalities are
\begin{equation}
		\begin{cases}
			\varepsilon_d\leq
			\frac{2-(M-1)\varepsilon_d}{2}\sum\limits_{i=1}^M\varepsilon_i \\
			\varepsilon_d\leq
			\sum\limits_{i=1}^M\varepsilon_i
			+\frac{M-1}{2M}\left(\sum\limits_{i=1}^M\varepsilon_i\right)^2.
		\end{cases}
	\end{equation}
Choosing the intersection
of inequalities, we obtain
a reliable upper bound of $\varepsilon_d '$,
\begin{equation}
		\varepsilon_d '\leq
		\frac{M-\sqrt{M^2-2M(M-1)\varepsilon_d}}{M-1}.
	\end{equation}
Thus, Lemma \ref{LEMMA 1} is proved.

\section{PROOF OF LEMMA \ref{LEMMA 2}}\label{sec:appendices 2}
We prove this lemma by contradiction. First, we assume that there is a set of solutions $\left\{x_1[1], x_2[2], \cdots, x_N[N]\right\}$ to the optimization problem, and the solution satisfies
\begin{equation}
		\sum\limits_{k=1}^N\left(\sum\limits_{m=k}^{N}Q(x_m[k])+Q(x_D[k])\right )<\varepsilon_d '  .
	\end{equation}
Then, because $Q(x)$ is a continuous function, $ \Delta x>0$ satisfies
\begin{equation}
		\begin{split}
			&\sum\limits_{k=1}^{N-1}\left(\sum\limits_{m=k}^{N}Q(x_m[k])+Q(x_D[k])\right )\\
			&\qquad+Q\left(x_N[N]-\Delta x\right) +Q\Big(a\left(x_N[N]-\Delta x\right)+b\Big)<\varepsilon_d ' .
		\end{split}
	\end{equation}
Here, the parameters $a=\sqrt{V_{S, N}/V_{N, D}}$ and $b=(C_{N, D}-C_{S, N})\sqrt{n/V_{N, D}}$.
Therefore, we have another set of solutions $\{x_1[1], x_2[2], \cdots, x_N[N]-\Delta x\}$ that satisfies the constraint. Its equivalent is less than the value of the assumed solution, which
contradicts our supposition.

Additionally, we have
\begin{equation}
		\sum\limits_{k=1}^N\left(\sum\limits_{m=k}^{N}Q(0)+Q(0)\right)>\varepsilon_d ',
	\end{equation}
Consequently, $\varepsilon_d '$ is an achievable point due to the intermediate value theorem.
The set of solutions must satisfy
\begin{equation}
		\sum\limits_{k=1}^N\left(\sum\limits_{m=k}^{N}Q(x_m[k])+Q(x_D[k])\right )=\varepsilon_d ',
	\end{equation}
which completes the proof.

\section{PROOF OF LEMMA \ref{LEMMA 3}}\label{sec:appendix 3}
The proof is based on the following observation.
To easily describe it, we denote the equality objective functions and constraint function of the original
optimization problem (\ref{eq:simplest  constraint}) by $f_0(\boldsymbol{x}) $ and $h_0(\boldsymbol{x})$, respectively. Similarly, the equality objective functions and constraint function of the optimization problem (\ref{eq:final constraint}) are $f_1(\boldsymbol{x}), h_1(\boldsymbol{x})$. The relationship between $f_0(\boldsymbol{x}), h_0(\boldsymbol{x})$ and
$f_1(\boldsymbol{x}), h_1(\boldsymbol{x})$ is expressed as
\begin{align}
		f_1(\boldsymbol{x})=h_0(\boldsymbol{x})+\varepsilon_d ', \\
		h_1(\boldsymbol{x})=f_0(\boldsymbol{x})-v_0 .
	\end{align}
Based on problem (\ref{eq:simplest  constraint}), we obtain the Karush-Kuhn-Tucker (KKT) conditions:
\begin{equation}
		\begin{cases}
			\nabla f_0(\boldsymbol{x_0})+\lambda \nabla h_0(\boldsymbol{x_0})=0, \\
			h_0(\boldsymbol{x_0})=0.                                   .         \\
		\end{cases}
	\end{equation}
Because $\nabla f_0(\boldsymbol{x})\neq 0, $ for all $\boldsymbol{x}$, it follows that $\lambda \neq 0$.
As a result, by substituting $\boldsymbol{x_0}$ and $1/\lambda$ into the KKT condition of the changed optimization problem, we have
\begin{align}
		\begin{split}
			\nabla f_1(\boldsymbol{x_0})+\frac{1}{\lambda} \nabla h_1(\boldsymbol{x_0})
			=&\nabla h_0(\boldsymbol{x_0})+\frac{1}{\lambda}(\nabla f_0(\boldsymbol{x_0}))\\
			=&\frac{1}{\lambda}(\nabla f_0(\boldsymbol{x_0})+\lambda \nabla h_0(\boldsymbol{x_0}))\\
			=&0,
		\end{split}\\
		\begin{split}
		h_1(\boldsymbol{x_0})=&f_0(\boldsymbol{x_0})-v_0\\
		=&0.
		\end{split}
	\end{align}
Therefore, $(\boldsymbol{x_0}, \frac{1}{\lambda})$ is a solution to the
optimization problem (\ref{eq:simplest  constraint}). Furthermore, because the
KKT condition of the optimization problem (\ref{eq:simplest  constraint})
consists of $N+1$ independent equations, there is only one solution satisfying the KKT condition. Hence, the solution to the optimization problem (\ref{eq:simplest  constraint}) is also $\boldsymbol{x_0}$, which completes the proof.

\ifCLASSOPTIONcaptionsoff
  \newpage
\fi



\bibliographystyle{IEEEtran}
\bibliography{CAO-SIR-FBC.bib}
%

%

\end{document}